\documentclass[prl,aps,twocolumn,amsfonts,showpacs,superscriptaddress]{revtex4-1}

\usepackage{graphicx}
\usepackage{xcolor}
\usepackage{rotating}
\usepackage{amsmath,amssymb,graphics,amsthm}

\usepackage[colorlinks=true, urlcolor=violet, linkcolor=blue, citecolor=red, hyperindex=true, linktocpage=true]{hyperref}

\newtheorem{thm}{Theorem}

\newtheorem{lem}[thm]{Lemma}

\newtheorem{prop}[thm]{Proposition}

\def\BC{\mathbb{C}}

\def\CH{\mathcal{H}}

\def\CL{\mathcal{L}}

\def\CO{\mathcal{O}}

\def\BP{\mathbb{P}}

\def\BZ{\mathbb{Z}}
\allowdisplaybreaks
\newcommand{\AC}[1]{
{\color{blue}#1}
}
\begin{document}
\title{Finite speed of quantum scrambling with long range interactions}

\author{Chi-Fang Chen}
\affiliation{Department of Physics, Stanford University, Stanford CA 94305, USA}
\affiliation{Department of Physics, California Institute of Technology, Pasadena CA 91125, USA}

\author{Andrew Lucas}
\email{andrew.j.lucas@colorado.edu}
\affiliation{Department of Physics, Stanford University, Stanford CA 94305, USA}
\affiliation{Department of Physics, University of Colorado, Boulder CO 80309, USA}

\begin{abstract}
In a locally interacting many-body system, two isolated qubits, separated by a large distance $r$, become correlated and entangled with each other at a time $t \ge r/v$.  This finite speed $v$ of quantum information scrambling limits quantum information processing, thermalization and even equilibrium correlations. Yet most experimental systems contain long range power law interactions -- qubits separated by $r$ have potential energy $V(r)\propto r^{-\alpha}$.  Examples include the long range Coulomb interactions in plasma ($\alpha=1$) and dipolar interactions between spins ($\alpha=3$).     In one spatial dimension, we prove that the speed of quantum scrambling remains finite  for sufficiently large $\alpha$.    This result parametrically improves previous bounds, compares favorably with recent numerical simulations, and can be realized in quantum simulators with dipolar interactions.  Our new mathematical methods lead to improved algorithms for classically simulating quantum systems, and improve bounds on environmental decoherence in experimental quantum information processors.
\end{abstract}

\date{\today}

\maketitle

Almost five decades ago, Lieb and Robinson proved that spatial locality implies the ballistic propagation of quantum information \cite{liebrobinson}.   Intuitively defining a ``scrambling time" $t_{\mathrm{s}}(r)$ by the time at which an initially isolated qubit can significantly entangle with another a distance $r$ away, locality implies that $t_{\mathrm{s}}(r)\gtrsim r$.   This result has deep implications in physics.  Practical tasks such as information processing \cite{landahl} are possible due to a lack of rapid decoherence with a noisy environment, thermalization occurs locally \cite{kaufman} and equilibrium correlation functions fall off sufficiently rapidly \cite{hastings}.  If quantum information can only propagate with a finite speed, a classical computer can efficiently approximate early time quantum dynamics \cite{haah}. Despite the exponentially large Hilbert space in many-body quantum systems, quantum information processors with short-range interactions cannot become entangled with an infinite environment arbitrarily quickly \cite{eisert,bravyi}.  Lastly, emergent spacetime locality arising from microscopic quantum mechanics without manifest relativistic invariance may play a crucial role in understanding quantum gravity through the holographic correspondence \cite{maldacena}. 

However, the Lieb-Robinson theorem is not useful for a typical quantum information processor.  A qubit in an experimental device is usually a spin or atomic degree of freedom, or Josephson junction.  Such objects generically interact with long range interactions, and until now, whether locality of quantum scrambling necessarily persists in the presence of long range interactions has remained unclear.  In 2005, Hastings and Koma used the canonical Lieb-Robinson theorem to prove that when $\alpha>d$,  $t_{\mathrm{s}}(r) \gtrsim \log r$ \cite{hastings}; more recently, this bound has been improved for $\alpha>2d$ to $t_{\mathrm{s}}(r) \gtrsim r^{(\alpha-2d)/(\alpha-d)}$ \cite{maryland14, maryland18, yao, maryland19}.   If such bounds were tight, then insulating a quantum processor from its environment would be absolutely crucial.   Yet numerical simulations cast into doubt the tightness of these formal bounds: two groups have recently shown that $t_{\mathrm{s}}\gtrsim r$ in one dimensional models with $\alpha\gtrsim 1.5$ \cite{chen} or even $\alpha>1$ \cite{luitz}, depending on microscopic details.  

In this letter, we prove that $t_{\mathrm{s}}(r)\gtrsim r$ whenever $\alpha>3$, in all one dimensional models with power law interactions.  Our dramatic improvement over existing results is made possible by new mathematics \cite{us}: identities for unitary time evolution expanded as a sum over flexibly chosen equivalence classes of sequences of couplings. 

Our work has clear physical consequences.  Scrambling in dipolar spin chains \cite{vernac, mark} is hardly faster than in a spin chain with nearest neighbor interactions; hence, it should be far more efficient to simulate numerically \cite{maryland18, haah}.  Nor does decoherence seriously limit the quantum information processing capabilities of a nuclear spin chain, no matter how large the environment.  Quantum thermalization nearly proceeds as if interactions were local, as in typical theoretical models of scrambling \cite{nahum, tibor}. 

\emph{Formal Statement of Theorem.---}  We now formally restate our theorem in a mathematically precise language.  For simplicity, we will assume a one dimensional chain of qubits (two-level systems); the generalization to all finite-dimensional quantum models in one dimension is contained in the Supplementary Material \cite{suppmat}.  Thus, the Hilbert space is given by \begin{equation}
\mathcal{H} = \bigotimes_{i \in \mathbb{Z}} \CH_i = \bigotimes_{i \in \mathbb{Z}} \BC^2.  \label{eq:tensorproduct}
\end{equation}
Even though $\mathcal{H}$ is (uncountably) infinite dimensional, our bound on scrambling will reduce to a calculation on an finite segment of the chain.

The set of Hermitian operators on $\mathcal{H}$ forms a real vector space $\mathcal{B}$.  Let the $\mathrm{U}(2)$ generators $\lbrace I, \sigma^x,\sigma^y, \sigma^z \rbrace $ be our complete basis of Hermitian operators on $\CH_i$; $\mathcal{B}$ is spanned by tensor products $ \{I, \sigma_i, \sigma_i\sigma_j, \cdots \}$
Here and below, we can use a ``bra-ket" notation with parentheses to emphasize that Hermitian operators on $\mathcal{H}$ are vectors in $\mathcal{B}$.  We define $\lVert A\rVert$ as the maximal eigenvalue of $A$, the conventional operator norm  \cite{liebrobinson}.

We consider 2-local Hamiltonians: i.e., those which may be expressed as a sum of terms which act on either a single site, or on two sites: \begin{equation}
H = \sum_{i \in \mathbb{Z}} H_i + \sum_{i<j} H_{ij}. \label{eq:H}
\end{equation}
 We define the exponent $\alpha$ of long range interactions by demanding that \begin{equation}
\lVert H_{ij}\rVert \le \frac{h}{|i-j|^\alpha}. \label{eq:hijbound}
\end{equation}

Let $i<j$ be integers.  We define the scrambling time $t_{\mathrm{s}}^\delta(r)$ to be the largest time such that 
\begin{equation}
\sup_{A_{\le i}, B_{\ge j}}\frac{\lVert [A_{\le i}(t'), B_{\ge j}] \rVert}{\lVert A_{\le i}\rVert \lVert B_{\ge j}\rVert} < \delta, \text{ for  } 0<|t'|<t_{\mathrm{s}}^\delta(|i-j|). \label{eq:scrambling}
\end{equation} 
where $A_{\le i}$ denotes a bounded operator that acts trivially on any site $k>i$ and $B_{\ge j}$ acts trivially on any site $k<j$.  While $A_{\le i}$ can act non-trivially on an infinite number of sites, we demand $B_{\ge j}$ acts non-trivially only on a finite number of sites \cite{suppmat}. Lastly, the operator $A_{\le i}(t^\prime) := \mathrm{e}^{\mathrm{i}Ht^\prime} A_{\le i} \mathrm{e}^{-\mathrm{i}Ht^\prime}$ denotes a Heisenberg time-evolved operator.  
The definition of scrambling given in (\ref{eq:scrambling})  bounds the growth in observable correlation functions, and the generation of entanglement between distant qubits \cite{eisert, bravyi}.  We are now ready to state our main result:

\begin{thm}  For every $0<\delta<2$, there exists a constant $0<K_{\alpha} <\infty$ for which \begin{equation}
t_{\mathrm{s}}^\delta(r) \ge K_\alpha \times \left\lbrace\begin{array}{ll} r^{\alpha-2} &\ 2<\alpha<3 \\ r(\log r)^{-2} &\ \alpha = 3 \\ r &\ \alpha>3 \end{array}\right.. \label{eq:main}
\end{equation}
\end{thm}

\emph{Sketch of Proof.---} We now outline the proof of Theorem 1; details are found in \cite{suppmat}. For simplicity, we set $i=1$ and $j=r$ in (\ref{eq:scrambling}).   We write $A_1$ and $B_r$ below as shorthand for $A_{\le 1}$ and $B_{\ge r}$.

In the Heisenberg picture of quantum mechanics, operators evolve according to $\partial_t \mathcal{O} = \mathrm{i}[H,\mathcal{O}]$.  Just like the Schr\"odinger equation, this is \emph{linear}: we write $\partial_t |\mathcal{O}) = \mathcal{L}|\mathcal{O})$ where $\mathcal{L}$, commutation with Hamiltonian $H$, generates time translations on the space of operators.   The time evolved operator $|\mathcal{O}(t)) = \mathrm{e}^{\mathcal{L}t}|\mathcal{O})$ is nothing more than a ``rotated" operator of the same norm.
We define the projection $\mathbb{P}_r$ onto the hyperplane $\Sigma_r$ of $\mathcal{B}$ of all operators that act non-trivially on the support of $B_r$. This is a convenient object that bounds scrambling by the evolution of $|A_1)$ into $\Sigma_r$ as a function of time: \begin{equation}
\frac{\lVert [A_1(t),B_r]\rVert}{2\lVert A_1\rVert \lVert B_r\rVert} \le \frac{\lVert \mathbb{P}_r\mathrm{e}^{\mathcal{L}t}|A_1)\rVert }{\lVert A_1\rVert }.
\end{equation}

In the ``canonical" form of the Lieb-Robinson theorem popularized by Hastings and Koma \cite{hastings}, one uses the triangle inequality:  $ \partial_t \lVert [A_1(t),B_r] \rVert \le \lVert [A_1(t), [H,B_r]] \rVert $.  Yet most of the terms on the right hand side of this inequality sum do not contribute to $\lVert[A_1(t),B_r]\rVert$: they correspond to shifts in $A_1(t)$ that cannot grow $\lVert \mathbb{P}_r |A_1)\rVert$.  We emphasize that this holds even though the operator norm $\lVert A_1(t)\rVert$ is not the ``length" of the vector $|A_1(t))$.  

\begin{figure*}
\includegraphics[width=\textwidth]{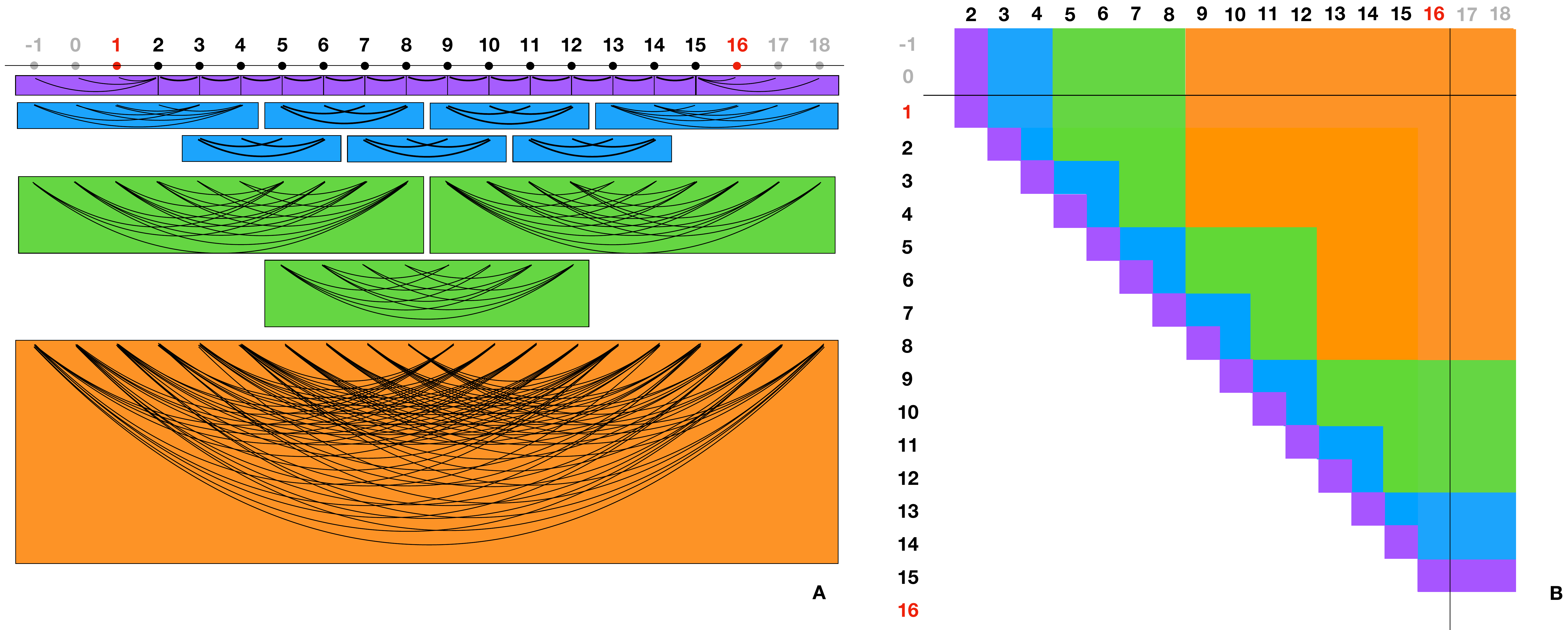}
\caption{Couplings $\mathcal{L}_{ij}$ ($i<j$) can be broken up into the \emph{scale} on which the coupling acts in a unique way.  Intuitively, the scale $q$ of a coupling is approximately $\lfloor \log_2 (j-i) \rfloor$.  Each scale with different values of $q$ is denoted with a different color: from large (orange) to short (purple).  In this example, we study $\lVert[A_1(t),B_{16}]\rVert$, and sites $n$ obeying $n<1$ or $n>16$ are grouped in with these end sites when combining couplings. (A) Each scale/color would form a chain; (B) A more precise presentation. This $L$-shaped tiling ensures that the sum of each $q$-block scales as $2^{-q(\alpha-2)}$ and can be extended to arbitrary large $q$. }  
\label{fig:groupbyscale}
\end{figure*}

Instead, we write \begin{equation}
\mathrm{e}^{\mathcal{L}t} |A_1) = \sum_{n=0}^\infty \frac{t^n}{n!} \sum_{X_1,\ldots, X_n} \mathcal{L}_{X_n}\cdots \mathcal{L}_{X_1}|A_1) \label{eq:Lseq}
\end{equation}
where $X_i$ corresponds to a term in the Hamiltonian: e.g. $\mathcal{L}_{\sigma^x_1 \sigma^x_2} = \mathrm{i}[\sigma^x_1\sigma^x_2,\cdot]$.  $\BP_r\mathcal{L}_{X_n}\cdots\mathcal{L}_{X_1}|A_1)$ is only non-zero if a subsequence of $\mathcal{L}$s form a path from 1 to $r$.  Our main technical development is expanding $\mathbb{P}_r\mathrm{e}^{\mathcal{L}t}|A_1)$ in a controlled way: we classify all sequences with a path from 1 to $r$ by a relatively small number of equivalence classes $\Gamma$. Generalizing the interacting picture, we obtain the following identity:
\begin{align}
\mathbb{P}_r \sum_{\Gamma} \sigma(\Gamma) &\int\limits_0^t \mathrm{d}t_{\ell}\int\limits_0^{t_\ell} \mathrm{d}t_{\ell-1}\cdots \int\limits_0^{t_2} \mathrm{d}t_{1}\mathrm{e}^{\mathcal{L}(t-t_{\ell})} \mathcal{L}^\Gamma_{\ell} \mathrm{e}^{\tilde{\mathcal{L}}^\Gamma_{\ell}(t_{\ell}-t_{\ell-1})} \notag \\
& \cdots\mathrm{e}^{\tilde{\mathcal{L}}^\Gamma_2(t_2-t_1)} \mathcal{L}^\Gamma_1 \mathrm{e}^{\tilde{\mathcal{L}}^\Gamma_1t_1}|A_1) = \mathbb{P}_r\mathrm{e}^{\mathcal{L}t}|A_1). \label{eq:mainequivclass1}
\end{align}
Here $\Gamma$ is a label for $\ell$ non-trivial sequential steps $\mathcal{L}^\Gamma_j$ and the time-ordered integral can intuitively be interpreted as the possible times $t>t_\ell > \cdots > t_1 $ at which ``critical" steps in the sequence of $\mathcal{L}_{X_i}$ occurred.  In fact, the emergence of the integral over $\ell$ ordered times is analogous to the time ordered integrals which arise in time dependent perturbation theory.  According to rules we will shortly state, $\sigma(\Gamma)=\pm 1$ is assigned to avoid double counting so that the terms match up across the equality.
 Applying the triangle inequality to (\ref{eq:mainequivclass1}), and noting $\mathrm{e}^{\tilde{\mathcal{L}}^\Gamma_j t}$ is norm-preserving, which resums superfluous terms in the series expansion (\ref{eq:Lseq}): \begin{equation}
\frac{\lVert \mathbb{P}_r\mathrm{e}^{\mathcal{L}t}|A_1) \rVert }{2\lVert A_1\rVert} \le \sum_\Gamma \frac{t^\ell}{\ell!} \prod_{j=1}^\ell \lVert \mathcal{L}^\Gamma_j\rVert . \label{eq:mainequivclass2}
\end{equation}
where 
 \begin{equation}
\lVert \mathcal{L}^\Gamma_j \rVert := \sup_{\mathcal{O}} \frac{\lVert \mathcal{L}^\Gamma_j  \mathcal{O}\rVert}{\lVert  \mathcal{O}\rVert} 
\label{supernorm}
\end{equation}

The physical content of (\ref{eq:mainequivclass2}) is interpreted as follows.  If we can define an equivalence relation on sequences of couplings such that (\ref{eq:mainequivclass1}) holds, then \emph{only the non-trivial steps} $\mathcal{L}_j^\Gamma$ need to be counted in the commutator bound (\ref{eq:scrambling}).  \emph{Every} other term in the sequence $\tilde{\mathcal{L}}^\Gamma_i$ that shows up in the intermediate unitary evolution does not grow $\lVert \mathbb{P}_r \mathrm{e}^{\mathcal{L}t}|A_1)\rVert$.

In general, the choice of $\Gamma$ is quite flexible.   For the Hamiltonian (\ref{eq:H}), our construction is depicted in Figure~\ref{fig:groupbyscale}.  We start by regrouping all $\mathcal{L}_{mn} = \mathrm{i}[H_{mn},\cdot]$ by their scale $q\approx \lfloor \log_2|m-n|\rfloor$ (if $1\le m,n\le r$; the exact formula is in \cite{suppmat}). We write $H$ as a sum of one dimensional Hamiltonians, each consisting of terms of a given scale.  At scale $q=0$ ($q>0$), these blocks form one dimensionals model of nearest (next nearest) neighbor interactions between \emph{blocks} of sites, see Figure~\ref{fig:groupbyscale}A.  Which couplings are grouped into which blocks is depicted in Figure~\ref{fig:groupbyscale}B.   At scale $q$, we denote the block $(q,k)$ to be the $k^{\mathrm{th}}$ left most block in Figure~\ref{fig:groupbyscale}, starting with $k=0$.   We denote $\mathcal{L}_{(q,k)} = \mathrm{i}[H_{(q,k)},\cdot]$ where \begin{equation}
H_{(q,k)} = \sum_{(i,j) \text{ in block } (q,k)} H_{ij}.
\end{equation}
We now rewrite (\ref{eq:Lseq}) as
 \begin{align}
\BP_r&\mathrm{e}^{\mathcal{L}t} |A_1) = \sum_{n=0}^\infty \frac{t^n}{n!} \times \notag \\
& \sum_{(q_n,k_n),\ldots, (q_1,k_1)}\BP_r \mathcal{L}_{(q_n,k_n)}\cdots \mathcal{L}_{(q_1,k_1)}|A_1) \label{eq:Lmnseq}
\end{align}

The key observation is that any sequence above (\ref{eq:Lmnseq}) that made it to $r$ must traverse forward a distance $\gtrsim r/\log r$ on \emph{at least one of the $\lfloor \log_2 r \rfloor$ scales $q$} (Figure \ref{fig3_longforward}). 
For any sequence of $\mathcal{L}_{(q_i,k_i)}$s in (\ref{eq:Lmnseq}) , we can read off the \textit{q-forward subsequence} $\mathbb{P}_r \cdots \mathcal{L}_{(q,k_N)}\cdots \mathcal{L}_{(q,k_1)}\cdots |A_1)$, $k_{N}>\cdots >k_{1}$ by recursively finding the next $(q,k_j)$ that exceeds the largest $k$ so far. If $N > N_{q} \sim 2^{-q} r/\log_2 r$, then the sequence is `'long".
We organize equivalence classes $\Gamma$ by the non-empty subset of the integers $\lbrace 0,1,\ldots, \lfloor \log_2 r \rfloor\rbrace$ that corresponds to the scales on which a ``long" path from 1 to $r$ exists.
    
For example, suppose a sequence in (\ref{eq:Lmnseq}) at scales $q_0$ contains a long subsequence. Then it would be accounted by the $\Gamma$ 
 specified by the first $N_{q_0}$ terms of the forward subsequence $0 \le k_1<k_2<\cdots <k_{N_{q_0}}<\cdots$. In (\ref{eq:mainequivclass1}), we take $\sigma(\Gamma)=1$, $\mathcal{L}^\Gamma_m = \mathcal{L}_{(q,k_m)}$, \begin{equation}
\mathcal{L}_{-q} = \mathcal{L} - \sum_k \mathcal{L}_{(q,k)}
\end{equation} and \begin{equation}
\tilde{\mathcal{L}}^\Gamma_m =  \mathcal{L}_{-q}  + \sum_{\AC{k}\le k_{m-1}} \mathcal{L}_{(q,k)}. \label{eq:LGammammain}
\end{equation}
where the ``superfluous" terms $\tilde{\mathcal{L}}^\Gamma_m$ are exactly those terms that does not change the $q$-forward subsequence.  To be precise, the terms allowed between $\mathcal{L}^\Gamma_m=\mathcal{L}_{(q,k_m)}$ and $\mathcal{L}^\Gamma_{m-1}=\mathcal{L}_{(q,k_{m-1})}$ are any $\mathcal{L}_{(q^\prime, k^\prime)}$ with $q\ne q^\prime$, or any $\mathcal{L}_{(q,k^{\prime\prime})}$ with $k^{\prime\prime}\le k_{m-1}$.   The former are allowed because each scale $q$ is treated separately; the latter because they are, by construction, not traversing forward.

In general, it may be the case that a single sequence in (\ref{eq:Lmnseq}) contains multiple long paths on $m$ distinct scales.  The equivalence class $\Gamma$ is then labeled by the $m$ long sequences at $m$ distinct scales.  We generalize the construction of the previous paragraph, and set $\sigma(\Gamma)=(-1)^{1+m}$.   The inclusion-exclusion principle then guarantees that (\ref{eq:mainequivclass1}) does not over count the sequences which have multiple long subsequences, as in Figure~\ref{fig3_longforward}.   

\begin{figure}
\includegraphics[width=0.5\textwidth]{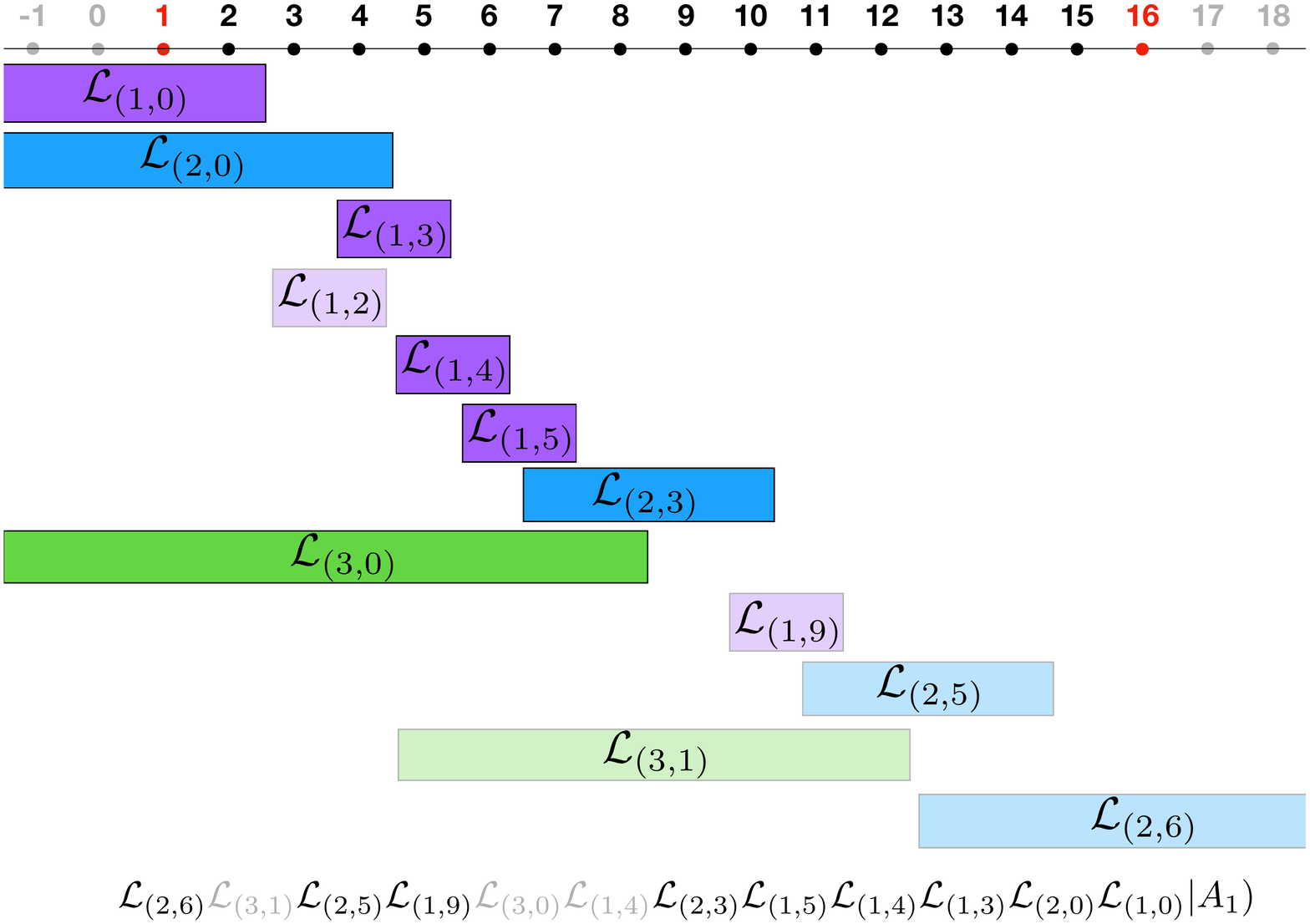}
\caption{Any sequence of $\mathcal{L}$ which grows $A_1(t)$ to the final site 16 must have a long sequence of couplings on at least one scale.  For the particular sequence shown, there are three scales with sufficiently long sequences (no shorter than $16/\log{16})=4$), and we bound the contribution of this sequence to $\lVert[A_1(t),B_{16}]\rVert$ by summing over the weight of all possible paths which contain the solid colored couplings (corresponding to $\mathcal{L}^\Gamma_j$) in a precise order.  The lightly shaded couplings (corresponding to $\tilde{\mathcal{L}}^\Gamma_j$) do not contribute to (\ref{eq:mainequivclass2}).}  
\label{fig3_longforward}
\end{figure}  

 In fact, in order to prove (\ref{eq:main}), we have improved this argument in a few ways.  (\emph{1}) We tune $N_q$ so that the contribution of all scales $q$ to (\ref{eq:mainequivclass2}) is comparable.  (\emph{2}) We demand that all ``long" paths must increase the right-most site on which the operator acts.   We then evaluate (\ref{eq:mainequivclass2}), using that $\lVert \mathcal{L}_{(q,k)}\rVert \lesssim 2^{-q(\alpha-2)}$.  Our results are summarized below.

When $\alpha>3$, the dominant contribution to $\lVert \mathbb{P}\mathrm{e}^{\mathcal{L}t}|A_1)\rVert$ comes from short length scales: a large fraction of the path from 1 to $r$ often occurs in \emph{nearest neighbor} hops.  Scrambling proceeds as if interactions were nearest neighbor alone.  The operator $|A_1(t))$ is largely supported on lattice sites $x<vt$, where $v$ is a finite speed of quantum scrambling.  

When $\alpha<3$, the dominant contribution to $\lVert \mathbb{P}\mathrm{e}^{\mathcal{L}t}|A_1)\rVert$ comes from few long hops across $1$ to $r$. Counting the number of such long hops, we find $t_{\mathrm{s}}(r) =\mathrm{O}(  r^{\alpha-2})$.

If $\alpha=3$, we find that all scales are equally important, which leads to $t_{\mathrm{s}}(r) =\mathrm{O}( r/\log^2 r)$.

 A final comment is that if the Hamiltonian is \emph{frustrated}, we may replace $\alpha \rightarrow \alpha -1$ in Theorem 1:  namely, the linear light cone persists until $\alpha=2$.  A formal definition of frustration is that the maximal eigenvalue of $H_{(q,k)}$ is comparable to its magnitude in a randomly chosen state.   This property is expected to hold for a self-averaging Hamiltonian where each $H_{ij}$ is multiplied by a zero-mean random variable.  Frustration does not hold in a Hamiltonian where all 2-local terms $H_{ij}$ in the Hamiltonian commute (e.g. $H_{ij} = \sigma^z_i \sigma^z_j / |i-j|^\alpha$).  

 \emph{Outlook.---}  We conclude the letter with a discussion of the implications of our theorem.   Recall that our new mathematical methods led to dramatic improvements over existing literature, where the previous optimal bound on scrambling in one dimensional systems was $t_{\mathrm{s}}(r)\gtrsim r^{(\alpha-2)/(\alpha-1)}$ for $\alpha>2$ \cite{maryland18}.  In fact, for any $\alpha>3$, the speed of quantum scrambling is finite:  entanglement \cite{eisert, bravyi} and quantum state transfer \cite{landahl} proceed at a finite rate, and thermalization largely mimics that of a locally interacting system.
 
Our results for frustrated systems are very similar to the numerical simulations of \cite{chen}, where it was argued that a finite speed of scrambling arises for $\alpha \gtrsim 1.5$ in a model with time-dependent random Hamiltonian.   However, in another model with fixed Hamiltonian \cite{luitz}, it was found that $\alpha \gtrsim 1$ marked the onset of the finite scrambling speed.  We conjecture that (\ref{eq:main}) holds with $\alpha \rightarrow \alpha - 1$ (hence, the light cone persists to $\alpha=2$) for all models, including those which are not (by our definition) frustrated, whenever the operators in (\ref{eq:scrambling}) act on a single site.  It would be interesting if this can be proved rigorously.

The techniques developed in this letter may generalize to other important problems in quantum information dynamics, including entanglement growth and quantum scrambling  in finite temperature thermal ensembles, where the speed of quantum scrambling averaged over the thermal ensemble may depend on temperature \cite{swingle}.  We also hope to generalize our main theorem to any spatial dimension $d$. Lastly, we have also used similar techniques to constrain models of holographic quantum gravity \cite{us}.  Given the recent explosion of interest in realizing analogue black holes in quantum simulators \cite{Garttner2017, Li2017}, our methods will constrain which experimental systems have the potential to achieve this ambitious goal.

\emph{Acknowledgements.---}   We thank Alexey Gorshkov, Andrew Guo and Minh Tran for pointing out an error in a previous version of the paper. This work was supported by the Gordon and Betty Moore Foundation's EPiQS Initiative through Grant GBMF4302.

\onecolumngrid
\pagebreak

\setcounter{equation}{0}
\setcounter{figure}{0}
\setcounter{thm}{0}
\renewcommand{\theequation}{S\arabic{equation}}
\renewcommand{\thefigure}{S\arabic{figure}}

\section*{Supplementary Material}

The supplementary material to this letter contains the formal proof of the following theorem:
\begin{thm} Define the parameter $\alpha^\prime$ as follows: \begin{equation}
\alpha^\prime = \left\lbrace \begin{array}{ll} \alpha &\ \text{$H$ is frustrated}  \\ \alpha- 1 &\ \text{otherwise} \end{array}\right..  \label{eq:alphaprime}
\end{equation}  
Let 
\begin{equation}
\mathcal{H} = \bigotimes_{i \in \BZ} \CH_i
\end{equation}
with $\dim(\CH_i)< \infty$.  Then for every $0<\delta<2$, there exists a constant $0<K_{\alpha} <\infty$ for which \begin{equation}
t_{\mathrm{s}}^\delta(r) \ge K_{\alpha^\prime} \times \left\lbrace\begin{array}{ll} r^{\alpha^\prime -1} &\ 1<\alpha^\prime < 2 \\ r(\log r)^{-2} &\ \alpha^\prime = 2 \\ r &\ \alpha^\prime>2 \end{array}\right..
\end{equation}
\end{thm}
The notation follows that introduced in the main text.

\emph{Proof of Theorem 1:}   
The first step is to massage the commutator norm $\lVert [A_{\le i}(t),B_{\ge j}]\rVert$ to a simpler form, along with simplifying the notation.
Without loss of generality, we may take the starting vertex to be on the left $i<j$.  
Define 
\begin{equation}
n_* = \left\lfloor \log_2 |j-i| \right\rfloor,
\end{equation}
define
\begin{equation}
R = 2^{n_*}
\end{equation}
and \begin{equation}
i^\prime = j - R+1.
\end{equation}
This last equation is used to push the starting vertex $i$ farther forward than it actually is. It suffices to bound the more general case $\lVert [A_{\le i'}(t),B_{\ge j}]\rVert$, since $A_{\le i}$ is an instance of $A_{\le i'}$. 

Let $\mathcal{B}_i$ be the set of all traceless Hermitian operators acting on $\mathcal{H}_i$.  A basis  for $\mathcal{B}$ is spanned by the Hermitian operators contained in the following sets: \begin{equation}
\mathcal{B} = \bigoplus_{S\subseteq \mathbb{Z}} \mathcal{B}_S = \bigoplus_{S\subseteq \mathbb{Z}} \bigotimes_{i\in \mathbb{Z}} \left\lbrace\begin{array}{ll} \mathcal{B}_i &\ i\in S\\ I_i &\ i\notin S \end{array}\right..
\end{equation} 
Here $I_i$ denotes the identity operator acting on $\mathcal{B}_i$.   We define the projection superoperator $\mathbb{P}_S$ (for $S\subseteq \mathbb{Z}$) as follows, by its action on the basis set above:
\begin{equation}
\mathbb{P}_S |\CO_{S^\prime} ) =\left\lbrace
\begin{array}{ll} |\mathcal{O}_{S^\prime}) &\ S\cap S^\prime \ne \emptyset \\ 0 &\ S\cap S^\prime = \emptyset \end{array}\right., \;\;\;\;\; \text{ for } \mathcal{O}_{S^\prime}\in\mathcal{B}_{S^\prime}.
\end{equation}
Let $B_{\ge j} = B_{[j,W]}$ be an operator supported on the set $[j,W]$ for some finite $W \in \mathbb{Z}$.
\begin{prop} If $\mathcal{O}$ is a Hermitian operator, 
\begin{equation}
\sup_{B_{[j,W]}} \frac{\lVert [\mathcal{O},B_{[j,W]}]\rVert}{\lVert \mathcal{O}\rVert \lVert B_{[j,W]}\rVert } \le 2 \frac{\lVert \mathbb{P}_{[j,W]}\mathcal{O}\rVert}{\lVert \mathcal{O}\rVert}
\end{equation}
\label{CommToProj}
\end{prop}

\begin{proof}
Starting with $[\mathcal{O},B_{[j,W]}] = [\BP_{[j,W]}\mathcal{O},B_{[j,W]}]$, by submultiplicativity and triangle inequality\begin{equation}
\frac{\lVert[\CO,B_{[j,W]}]\rVert}{\lVert B_{[j,W]}\rVert} =\frac{\lVert[\BP_{[j,W]}\CO,B_{[j,W]}]\rVert}{\lVert B_{[j,W]}\rVert} \le 2 \lVert \BP_{[j,W]}\CO\rVert
\end{equation} 
\end{proof}

\begin{prop}
$\mathbb{P}_{[j,W]}$ cannot arbitrarily grow the norm of an operator:
\begin{align}
\lVert \mathbb{P}_{[j,W]}\mathcal{O}\rVert \le 2\lVert\mathcal{O}\rVert 
\end{align}
\label{normdecrease}
\end{prop}
\begin{proof}
Without loss of generality, let $\dim(\mathcal{H}_{[j,W]}) = d$, and let $d=2^m$.  If $d\ne 2^m$, we may replace $d$ with $2^{\lceil \log_2 d\rceil}$ and simply treat $H_{[j,W]}$ as vanishing on any of the added states.     Once $d=2^m$, we consider \begin{equation}
\mathcal{H}_{[j,W]} \simeq \bigotimes_{j=1}^m \mathbb{C}^2,
\end{equation}
and let $\sigma^\omega_j$ denote the Pauli matrices acting on each of the $m$ 2-dimensional Hilbert spaces.   A complete basis for all operators acting on $\mathcal{H}_R$ is spanned by\begin{equation}
|T^{\omega}) :=  \bigotimes_{j=1}^m \sigma^{\omega_j}_j, \;\;\; (\alpha_j=0,1,2,3),
\end{equation}
where $\sigma^0_j$ represents the identity matrix acting on block $j$.   Note that the operator norms $\lVert T^{\omega}\rVert = 1$ for all $4^m$ basis vectors.

Next, notice that the projector is proportional to the Casimir element of $\mathrm{SU}(d)$: 
\begin{equation}
\BP_{[j,W]} \CO = \frac{1}{2 \cdot 4^{m}} \sum_{\lbrace \omega \rbrace } [T^{\omega},[T^{\omega},\CO]].  \label{eq:improvedprop3}
\end{equation}
Then by (\ref{eq:improvedprop3}), $\lVert T^{\omega}\rVert = 1$, submultiplicativity and the triangle inequality, we obtain
\begin{equation}
\lVert \BP_{[j,W]} \CO \rVert \le \frac{4\times (4^m-1)}{2 \cdot 4^{m}} \lVert \CO \rVert < 2 \lVert \CO \rVert.
\end{equation}

For readers unfamiliar with (\ref{eq:improvedprop3}), we provide an explicit proof as follows. As this identity is linear, we only need to show for operators of the form $\mathcal{O}_{[j,W]} \otimes \mathcal{O}_{[j,W]^c}$, where $\mathcal{O}_{[j,W]^c}$ denotes an arbitrary operator acting on $\CH_{[j,W]^c}$.   Without loss of generality, we take $\mathcal{O}_{[j,W]}$ to be one of the basis vectors $T^{\omega}$.   

A simple calculation shows that two Pauli string $\omega, u$ either commute or anti-commute \begin{equation}
T^{\omega}T^{u} = S(\omega,u) T^{u}T^{\omega},
\end{equation}
where \begin{equation}
S(\omega, u) = \prod_{j=1}^m \left\lbrace \begin{array}{ll} 1 &\ \omega_j=0, \; u_j = 0 \text{, or }\omega_j=u_j \\ -1 &\ \text{otherwise} \end{array}\right..
\end{equation}
If $T^{\omega_j}$ is the identity (i.e. all $\omega_j=0$), every commutator in (\ref{eq:improvedprop3}) vanishes:  indeed such terms must vanish from the projector by definition; if there is any $\omega_j$ that does not equal zero, then exactly half ($\frac{4^m}{2}$) of the commutators in (\ref{eq:improvedprop3}) do not vanish. Observe that for any choice of $(u_2, u_3,\cdots, u_m),$ it produces two $ S = +1 $ and two $S=-1$.  \begin{equation}
S(\omega, u) = \left\lbrace \begin{array}{ll} -S(\omega_2,\ldots ,\omega_m, u_2,\ldots, u_m) &\ \omega_1\ne 0\text{ and } \omega_1\ne u_1 \\ S(\omega_2,\ldots ,\omega_m, u_2,\ldots, u_m) &\ \omega_1= 0\text{ or } \omega_1= u_1 \end{array}\right.
\end{equation} 
  Then (\ref{eq:improvedprop3}) follows from \begin{equation}
\sum_{\lbrace \alpha_j \rbrace } [T^{\alpha_j},[T^{\alpha_j},T^{\beta_j}]] = \sum_{\lbrace \alpha_j\rbrace} (1-S(\alpha_j,\beta_j))^2T^{\beta_j} = (4\times \frac{1}{2}\times 4^{m}) T^{\beta_j}.
\end{equation}
\end{proof}

By Proposition~\ref{CommToProj}, we can replace the commutator norm $\lVert [A_{\le i^\prime}(t),B_{[j,W]}]\rVert$ of Lieb and Robinson by the norm of a projection.   For simplicity, we now simplify the notation to mimic that of the main text:
 \begin{equation}
\lVert \mathbb{P}_{[j,W]}|A_{\le i'}(t)) \rVert= \lVert \mathbb{P}_{R}|A_{1}(t))\rVert
\end{equation}
where we have shortened the index $[j,W]$ by $R$ and the index $\le i^\prime$ by 1.   We expect that our results also hold when the operator $B_{[j,W]}$ acts on an infinitely large subspace, but this introduces additional mathematical complications which are otherwise unnecessary.  Indeed, $W$ will never show up again in our proof.  Note that $2R\ge r= |i-j|\ge R =  |i^\prime-j|$; our shortening of the domain of interest will not qualitatively modify our results.

The next step of the proof, as sketched in the main text, is to organize the sequences of Liouvillians $\mathcal{L}_{X_n}\cdots \mathcal{L}_{X_1}$ in (\ref{eq:Lseq}) by paths from $i^\prime$ to $j$ on multiple different scales.  Given two non-negative integers $q\ge 1$ and $k\ge 0$, we define the sets 
\begin{subequations}\begin{align}
\mathcal{Q}(1,k) &:= \lbrace \lbrace k+1,k+2 \rbrace \rbrace , \\
\mathcal{Q}(q,k) &:= \lbrace \lbrace m,n\rbrace  : 1\le m<n\le R , \;\;\; 2^{q-1} k < m < n \le 2^{q-1} (k+2)\rbrace  - \bigcup_{k^\prime \ge 0, q^\prime < q} \mathcal{Q}(q^\prime,k^\prime), \;\;\; (q>1). 
\end{align}\end{subequations}
These sets contain all the couplings (at each scale) which can propagate information forward, and will be used to reduce the problem to a simpler calculation on a one dimensional line with nearest neighbor interactions.  We reorganize the 2-local Liouvillians $\mathcal{L}_{ij}$ according to $\mathcal{Q}(q,k)$: \begin{equation}
\mathcal{L} = \sum_{(q,k)} \mathcal{L}_{(q,k)} \label{eq:LLqk}
\end{equation}
where \begin{equation}
 \mathcal{L}_{(q,k)}:= \sum_{\lbrace m,n\rbrace \in \mathcal{Q}(q,k)} \widetilde{\mathcal{L}}_{mn}.
\end{equation}
and we define the shifted 2-local Liouvillians to take care for interaction with longer that $|i^\prime-j|$ (e.g. $\mathcal{L}_{i^\prime-10, j+4}$):
 \begin{equation}
\widetilde{\mathcal{L}}_{mn} := \left\lbrace \begin{array}{ll} \mathcal{L}_{i^\prime+m-1,i^\prime+n-1} &\ 1<m<n<R \\ \displaystyle \sum_{k\le i^\prime} \mathcal{L}_{k,i^\prime+n-1} &\ 1=m<n<R \\ \displaystyle \sum_{k\ge j} \mathcal{L}_{i^\prime+m-1, k} &\ 1<m<n=R \\ \displaystyle \sum_{k\le i^\prime, k^\prime \ge j} \mathcal{L}_{k,k^\prime} &\ 1=m<n=R  \end{array}\right..
\label{eq:edgecase}
\end{equation}

One definition of a frustrated Hamiltonian is that there exists a constant $K$ such that for all $(q,k)$:
\begin{equation}
K \lVert H_{(q,k)}\rVert\le \lVert H_{(q,k)}\rVert_2,
\label{frustrated}
\end{equation}
with constant $0<K<\infty$ independent of $q$, and \begin{equation}
\lVert H_{(q,k)}\rVert_2^2 = \frac{\mathrm{tr}(H_{(q,k)}^2)}{\dim(\mathcal{H}_{(q,k)})}.
\end{equation}

\begin{lem} \label{renormalize}
The super-operator norm is bounded by 
 \begin{equation}
\lVert \mathcal{L}_{(q,k)} \rVert_{} \le  \frac{b}{2^{q(\alpha^\prime-1)}}  
\label{Leff}
\end{equation}
where
\begin{equation}
b :=  h\times\left\lbrace \begin{array}{ll} \displaystyle \dfrac{2^{2\alpha-\frac{1}{2}}}{(\alpha-1)K} &\ \text{frustrated model} \\\displaystyle \dfrac{2^{\alpha+2}}{(\alpha-1)(\alpha-2)} &\ \text{any $H$}  \end{array}\right..
\end{equation}

\end{lem}

\begin{proof}\emph{Case 1: Frustrated models}.  Observe that \begin{align}
\lVert \mathcal{L}_{(q,k)} \rVert &
 = \sup_{\mathcal{O}} \frac{\lVert \mathcal{L}_{(q,k)}  \mathcal{O}\rVert}{\lVert  \mathcal{O}\rVert} \le 2 \lVert H_{(q,k)} \rVert = \frac{2}{K}\sqrt{\frac{\mathrm{tr}\left(H_{(q,k)}^2 \right)}{\dim(\mathcal{H}_{(q,k)})}}\\
 &
 \le\frac{2}{K}\sqrt{\sum_{\lbrace m,n \rbrace \in \mathcal{Q}(q,k)} \lVert \widetilde{H}_{mn}\rVert^2 } \notag \\ 
&\le \frac{2}{K} \sqrt{\sum_{m=2^q(k+\frac{1}{2})-1}^{-\infty} \sum_{n=2^q(k+1) }^\infty\frac{h^2}{|m-n|^{2\alpha}} + \sum_{m=2^q(k+1)-1}^{-\infty} \sum_{n=2^q(k+\frac{3}{2}) }^\infty \frac{h^2}{|m-n|^{2\alpha}} } \notag \\
&< \frac{2}{K}\sqrt{2}\sqrt{\sum_{m=2^q(k+\frac{1}{2})-1}^{-\infty} \sum_{n=2^q(k+1) }^\infty \frac{h^2}{|m-n|^{2\alpha}}} < \frac{2}{K}h\sqrt{2}\sqrt{\sum_{m,n=0}^\infty \frac{1}{|2^{q-1} + m+n|^{2\alpha}}} \notag \\
&< \frac{2}{K}h\sqrt{2}\sqrt{\sum_{m,n=1}^\infty \frac{3^{2\alpha}}{|2^{q-1} + m+n|^{2\alpha}}} < \frac{3^\alpha 2^{\frac{3}{2}}h}{K} \sqrt{\int\limits_0^\infty \mathrm{d}m\int\limits_0^\infty \mathrm{d}n \frac{1}{(2^{q-1}+m+n)^{2\alpha}}} \notag \\
&< \frac{3^\alpha 2^{\frac{3}{2}}h}{K}\sqrt{\int\limits_0^\infty \mathrm{d}m \frac{1}{(2\alpha-1)(2^{q-1}+m)^{2\alpha-1}}}  < \frac{3^\alpha 2^{\frac{3}{2}}h}{K\sqrt{(2\alpha-1)(2\alpha-2)} 2^{(q-1)(\alpha-1)}}
\end{align}
where in the first line, we used the triangle inequality and submultiplicativity ($\lVert A B\rVert \le \lVert A \rVert \lVert B\rVert$); in the second we used the fact that the product of two non-trivial two-body operators acting on non-identical degrees of freedom must be traceless; in the third line we constrained all possible pairs $\lbrace m,n \rbrace$ in $\mathcal{Q}(q,k)$ ; in the fourth line we employed (\ref{eq:hijbound}), and the remainder of inequalities are elementary manipulations.

\emph{Case 2: Any $H$}.  We simply use the triangle inequality on $\lVert H_{(q,k)}\rVert$:
\begin{align}
\lVert \mathcal{L}_{(q,k)} \rVert &\le 2 \lVert H_{(q,k)} \rVert  \le  2 \sum_{m=2^q(k+\frac{1}{2})-1}^{-\infty} \sum_{n=2^q(k+1) }^\infty\frac{h}{|m-n|^{\alpha}} + \sum_{m=2^q(k+1)-1}^{-\infty} \sum_{n=2^q(k+\frac{3}{2}) }^\infty \frac{h}{|m-n|^{\alpha}} \notag \\
&<2h \sum_{m=2^q(k+\frac{1}{2})-1}^{-\infty} \sum_{n=2^q(k+1) }^\infty \frac{1}{|m-n|^\alpha} < 4h \sum_{m,n=1}^\infty \frac{2^\alpha}{|2^{q-1}+m+n|^\alpha}  < \frac{2^{\alpha+2}h}{(\alpha-1)(\alpha-2)2^{(q-1)(\alpha-2)}}
\end{align}

\end{proof}

Let $\beta_i=(q,k)$ denote one of the sets of couplings at scale $q$ described  above.  For convenience, when $\beta=(q,k)$, we will write $q(\beta) = q$ and $k(\beta) = k$.   Let $(\beta_1,\ldots, \beta_n)$ denote an ordered sequence of Liouvillians $\mathcal{L}_{\beta_n}\cdots \mathcal{L}_{\beta_1}$.    

\begin{lem}
 Every non-vanishing sequence must satisfy 
 \begin{subequations}\begin{align}
&k(\beta_1) = 0\\
&2^{q(\beta_m)-1} k(\beta_m) + 1 \le \max_{1 \le m^\prime <m} \left( 2^{q(\beta_{m^\prime})-1}(k(\beta_{m^\prime})+2) \right).  \label{eq:1dcreeping}
\end{align}\end{subequations}
This kind of sequence $\beta = (\beta_1,\ldots, \beta_n)$ is an instance of a broader notion called creeping \cite{us} applied to this system. \label{lemmacreeping}
\end{lem}

\begin{proof} This proof also follows \cite{us} and is straightforward. $\mathcal{L}_{\beta_n}\cdots \mathcal{L}_{\beta_1}|A_1) \ne 0$ implies $\CL_{\beta_1}$ overlaps with site $A_1$, which implies $k(\beta_1)=0$. It also implies $\mathcal{L}_{\beta_m}$ overlap with $\mathcal{L}_{\beta_{m-1}}\cdots \mathcal{L}_{\beta_1}|A_1) $, which is the condition (\ref{eq:1dcreeping}).

%
\end{proof}

We say that a sequence $\beta = (\beta_1,\ldots,\beta_n)$ is a \emph{forward} sequence from $j_1$ to $j_2$ if for all $1\le m<n$, $2^{q(\beta_m)-1}(k(\beta_m)+2) < 2^{q(\beta_{m+1})-1}(k(\beta_{m+1})+2)$, and if $2^{q(\beta_1)-1}k(\beta_1) = j_1$ and $2^{q(\beta_{n})-1}(k(\beta_{n})+2)  = j_2$. As we will see in Lemma \ref{ExistForward}, every creeping sequence from 1 to $R$ must have a sufficiently ``long" forward subsequence, and these forward sequences will then play a crucial role in our proof.  We define \begin{equation}
N_q = \left\lceil \frac{1}{2} \dfrac{2^{-q(\alpha^\prime-2)/2}}{\displaystyle \sum_{q^\prime=1}^{n_*} 2^{-q^\prime(\alpha^\prime-2)/2}}  \frac{R}{2^{q}} \right\rceil \label{eq:Nq}
\end{equation}
to be the number of couplings at scale $q$ which makes a sequence ``long" -- in our context, we chose $N_q$ such that the long paths at each scale contributes to the commutator norm slowly and somewhat ``equally" between all scales.  (This will be proven towards the end of our proof of the theorem.) We say that a forward sequence from $i^\prime$ to $j$ is a \emph{long $q$-forward sequence}  from $1$ to $R$ if (\emph{1}) it contains a forward subsequence of length $N_q$, $\beta_q = (\beta_{i_1},\ldots,\beta_{i_{N_q}})$ with the same scale $q = q(\beta_{i_m})$ for all $1\le m\le N_q$, and (\emph{2}) any forward subsequence $\beta^\prime$ remains forward if any element of $\beta_q$ is added to the sequence $\beta^\prime$.  In simpler terms, this forward subsequence must correspond to a sequentially increasing sequence of couplings at scale $q$, each of which also can grow the operator to the right.
As a matter of bookkeeping, we denote subsequence $\beta'$ of $\beta$ as $\beta^\prime\subseteq \beta$  and define characteristic functions $\chi_q$ to indicate sequences with long $q$-subsequences:
\begin{equation}
\chi_q \mathcal{L}_{\beta_p}\cdots \mathcal{L}_{\beta_1}
:= \left\lbrace \begin{array}{ll}
\mathcal{L}_{\beta_p}\cdots \mathcal{L}_{\beta_1}  &\  \text{if there exists long $q$-forward subsequence } \beta^\prime\subseteq \beta \\
0 &\ \text{else} \end{array}\right..
\label{eq:GZdef}
\end{equation}  

Having proven the lemmas above, we now set the stage for the remainder of the proof.   Let $\mathcal{S}$ denote the set of all creeping sequences which contain a forward subsequence from 1 to $R$, \begin{equation}
\mathbb{P}_R \mathrm{e}^{\mathcal{L}t} |A_1) = \mathbb{P}_R \sum_{p=0}^\infty \frac{t^p}{p!}\mathcal{L}^p |A_1) = \mathbb{P}_R \sum_{p=0}^\infty \frac{t^p}{p!} \sum_{\beta \in \mathcal{S}:|\beta|=p} \mathcal{L}_{\beta_p}\cdots \mathcal{L}_{\beta_1}|A_1). 
 \label{eq:PReLt}
\end{equation}
Naively bounding (\ref{eq:PReLt}) would lead to a lousy bound. The main idea is that we can repackage these terms using the inclusion-exclusion principle, where each group of term resums nicely. We exclude the paths without long $q$-forward sequences for any $q$: such paths vanish, as they cannot creep far enough to reach $R$, as shown by the following lemma:
\begin{lem} If $\beta = (\beta_1,\ldots, \beta_n)$ is creeping and $\mathcal{L}_{\beta_n}\cdots \mathcal{L}_{\beta_1}|A_0) \ne 0$, then it has a long $q$-forward subsequence for at least one integer $0\le q \le n_*$.  
\label{ExistForward}
\end{lem}
\begin{proof}
 We proceed in two steps, first showing that we can always construct a (possibly empty) $q$-forward subsequence of any creeping $(\beta_1,\ldots, \beta_n)$, and secondly showing that at least one of the sequences must be large.

Firstly, we explicitly construct a $q$-forward subsequence $\beta^q\subseteq \beta$ as follows.  Start with an empty sequence $\beta^q =  ()$; then read the sequence $\beta$ in order.  If an $m$ at which $q(\beta_m)=q$ is found, and $(k(\beta_m)+2)2^{q(\beta_m)-1} > (k(\beta_{m^\prime})+2)2^{q(\beta_{m^\prime})-1} $ for any $m^\prime<m$, set $\beta^q = (\beta_m)$.  Afterwards, suppose that the current sequence $\beta^q$ terminates with coupling $\beta_{m_0}$ and that we have read $\beta$ up to coupling $m$.  If $q(\beta_m)=q$ and $(k(\beta_m)+2)2^{q(\beta_m)-1} > (k(\beta_{m^\prime})+2)2^{q(\beta_{m^\prime})-1} $ for all $m^\prime < m$, replace $\beta^q \rightarrow (\beta^q, \beta_m)$.  The final sequence $\beta^q$ which we obtain is the output of this algorithm.  By construction, this is a forward (sub)sequence made out of only $q$-scale couplings, so it is $q$-forward.  The sequence $\beta^q$ need not be creeping.

For a contradiction, suppose that none of the $q$-forward subsequences found above are long.   Let $\hat\beta$ be the maximal forward subsequence of $\beta$; note that $\beta^1 \cup \cdots \cup \beta^{n_*} = \beta$.  If the sequence crept all the way beyond $R$, then trivially we have 
 \begin{equation}
R < \sum_{p=1}^{\ell(\hat\beta)} 2^{q(\hat\beta_p)} .   \label{eq:lemma41}
\end{equation}
By definition, every coupling that shows up in the forward sequence $\hat\beta$ must show up in a $q$-forward sequence for some $q$, so \begin{equation}
\sum_{p=1}^{\ell(\hat\beta)} 2^{q(\hat\beta_p)} < \sum_{q=1}^{n_*} 2^{q}\ell(\beta^q). \label{eq:lemma43}
\end{equation}
Now, by assumption every $q$-forward subsequence $\beta^q$ had $\ell(\beta^q)<N_q$, and we arrive at a contradiction:\begin{equation}
R< \sum_{q=1}^{n_*} 2^{q}\ell(\beta^q) < \frac{1}{2} \dfrac{\displaystyle \sum_{q=1}^{n_*} 2^{q} \times \frac{R}{2^q} 2^{-q(\alpha-2)/2} }{\displaystyle \sum_{q^\prime=1}^{n_*} 2^{-q^\prime(\alpha-2)/2}} = \frac{R}{2}. \label{eq:lemma42}
\end{equation}
 
\end{proof}

The next step is to convert Lemma \ref{ExistForward} into an explicit identity of the form (\ref{eq:mainequivclass1}).
\begin{prop}
\begin{align}
 \mathbb{P}_R \sum_{p=0}^\infty \frac{t^p}{p!}\mathcal{L}^p |A_1)  &= \left[ 1 - \prod_{q=1}^{n_*} (1-\chi_q)  \right]\mathbb{P}_R \sum_{p=0}^\infty \frac{t^p}{p!}\mathcal{L}^p |A_1) \notag \\
&=\left[ \sum_q \chi_q - \sum_{q_1<q_2} \chi_{q_1} \chi_{q_2}+\cdots (-1)^k \sum_{q_1<q_2<...<q_k} \chi_{q_1} \chi_{q_2}\cdots \chi_{q_k}+\cdots\right]\mathbb{P}_R \sum_{p=0}^\infty \frac{t^p}{p!}\mathcal{L}^p |A_1) \notag \\
&= -\sum_{Z\ne \emptyset, Z\subset \{1, \cdots, n_* \}} (-1)^{|Z|} \prod_{q\in Z} \chi_q \cdot \mathbb{P}_R \sum_{p=0}^\infty \frac{t^p}{p!}\mathcal{L}^p |A_1). \label{eq:sumZ}
 \end{align}
 \end{prop}
\begin{proof}
For each sequence in $\prod_{q=1}^{n_*} (1-\chi_q)  \mathbb{P}_R \sum_{p=0}^\infty \frac{t^p}{p!}\mathcal{L}^p |A_1)$, if sequence $\mathcal{L}_{\beta_\ell}\cdots \mathcal{L}_{\beta_1}|A_1)$ is not creeping then it vanishes; if it is creeping then by Lemma \ref{ExistForward} it vanishes. Hence $ \mathbb{P}_R \prod_{q=1}^{n_*} (1-\chi_q)  \sum_{p=0}^\infty \frac{t^p}{p!}\mathcal{L}^p |A_1)= 0$.  In the second line of (\ref{eq:sumZ}) we simply expand the polynomial of $\chi_q$, and in the last line we simply rewrite the result.
\end{proof}

To bound $\chi_q \mathrm{e}^{\CL t}$, we now need to classify every term in $\chi_q \mathrm{e}^{Lt}$ by the \textit{irreducible q-forward sequence} $\beta = (\beta_1,\ldots, \beta_\ell)$, constructed as follows: run the constructive algorithm of Lemma~\ref{lemmacreeping} to find the $q$-forward subsequence $\beta^\prime\subseteq (\beta_1,\ldots,\beta_p)$, and then truncate the tail of $\beta'$ such that $\ell(\beta') = N_q$.   
We denote the set of irreducible $q$-forward sequences $\mathcal{F}_q$.  Sequences with the same irreducible $q$-forward sequence can be resummed as follows:

\begin{lem}
\begin{align}
\chi_q \mathrm{e}^{\CL t } |A_1)
= \sum_{\beta \in \mathcal{F}_q}\int\limits_{\mathrm{\Delta}^\ell(t)} \mathrm{d}t_\ell\cdots \mathrm{d}t_1 \mathrm{e}^{\mathcal{L}(t-t_\ell)} \mathcal{L}_{\beta_\ell} \mathrm{e}^{\mathcal{L}^\beta_\ell(t_\ell - t_{\ell-1})}\mathcal{L}_{\beta_{\ell-1}} \mathrm{e}^{\mathcal{L}^\beta_{\ell-1}(t_{\ell-1} - t_{\ell-2})} \cdots \mathcal{L}_{\beta_1} \mathrm{e}^{\mathcal{L}^\beta_1t_1} \mathcal |A_1) \label{eq:chiq}
\end{align}
where $\ell = \ell(\beta)$, \begin{equation}
\mathcal{L}^\beta_p := \mathcal{L} - \sum_{\lambda \in Y^q_p(\beta)} \mathcal{L}_\lambda \label{eq:qLbetap}
\end{equation}
with \begin{align}
Y^q_p(\beta) &:=  \lbrace (q^\prime, k ): (k+1)2^{q^\prime-1}\ge (k(\beta_p)+1)2^{q-1} \rbrace, 
\end{align}
and $\mathrm{\Delta}^\ell(t)$ denotes the $\ell$-simplex: \begin{equation}
\mathrm{\Delta}^\ell(t) := \lbrace (t_1,\ldots, t_\ell) \in [0,t]^\ell : t_1 \le t_2 \le \cdots \le t_\ell \rbrace
\end{equation}
with volume \begin{equation}
\int\limits_{\mathrm{\Delta}^\ell(t)} \mathrm{d}t_\ell\cdots \mathrm{d}t_1 = \frac{t^\ell}{\ell!}.
\end{equation}\label{lemma7}
\end{lem}

\begin{proof} This is proved mirroring the proof of Theorem 4 of \cite{us}.  
First, we show that \begin{equation}
\chi_q \mathrm{e}^{\CL t } = \sum_{\lambda \in \mathcal{F}_q} \sum_{m_0,\ldots, m_{\ell(\lambda)}=0}^\infty \frac{(t\mathcal{L})^{m_{\ell(\lambda)}} (t\mathcal{L}_{\lambda_{\ell(\lambda)}})(t\mathcal{L}^\lambda_{\ell(\lambda)})^{m_{\ell(\lambda)-1}}\cdots (t\mathcal{L}^\lambda_2)^{m_1} (t\mathcal{L}_{\lambda_1})(t\mathcal{L}^\lambda_1)^{m_0}}{(\ell(\lambda) + \sum_{j=0}^{\ell(\lambda)} m_j)!}   \label{eq:chiqpart1}
\end{equation}
with $\mathcal{L}^\beta_p$ defined in (\ref{eq:qLbetap}).  Every sequence on the right hand side of (\ref{eq:chiqpart1}) corresponds to a term on the left because each of these sequences contains a $\lambda \in \mathcal{F}_q$ and thus has a long $q$-forward subsequence. Next, every sequence on the left can be written as a sequence on the right: by construction, the $Y^q_p(\beta)$ sets of couplings are chosen so that $\mathcal{L}^\beta_p$ does not change the irreducible $q$-forward subsequence of the term. The uniqueness of irreducible q-forward path implies that every term on the right hand side shows up exactly once. As the coefficients of terms on both sides of (\ref{eq:chiqpart1}) are the same, and we have found a bijection between the terms on both sides of the proposed equality (\ref{eq:chiqpart1}), we have demonstrated its veracity.

Secondly, we invoke a ``generalized Schwinger-Karplus" identity proved in \cite{us}, which equates the right hand side of (\ref{eq:chiqpart1}) to the right hand side of (\ref{eq:chiq}).
\end{proof}

$\chi_{q_1} \chi_{q_2}\cdots \chi_{q_k} \mathrm{e}^{\CL t}$ can be understood by putting each $\chi_{q_i}$ together ``indepedently."  Indeed, we can  classify every term in $\chi_{q_1} \chi_{q_2}\cdots \chi_{q_k} \mathrm{e}^{\CL t}$ by the irreducible $q$-forward sequence at each scale $q$ relatively independently: the only extra data we need is how the sequences weave between each other (i.e., the relative orders of all couplings between the long $q$-forward sequences for $q\in Z$).   Defining $\mathcal{F}_{Z}$ as the set of sequences composed of the weaving together of $\beta^q \in \mathcal{F}_q, q\in Z$ (irreducible $Z$-forward sequences), we arrive at the following lemma:

\begin{lem} 
\begin{align}
\chi_{q_1} \chi_{q_2}\cdots \chi_{q_k} \mathrm{e}^{\CL t}|A_1)&= \prod_{q\in Z}\chi_q \cdot \mathrm{e}^{\CL t}|A_1)\notag\\
&= \sum_{\beta \in \mathcal{F}_Z}\int\limits_{\mathrm{\Delta}^\ell(t)} \mathrm{d}t_\ell\cdots \mathrm{d}t_1 \mathrm{e}^{\mathcal{L}(t-t_\ell)} \mathcal{L}_{\beta_\ell} \mathrm{e}^{\mathcal{L}^\beta_\ell(t_\ell - t_{\ell-1})}\mathcal{L}_{\beta_{\ell-1}} \mathrm{e}^{\mathcal{L}^\beta_{\ell-1}(t_{\ell-1} - t_{\ell-2})} \cdots \mathcal{L}_{\beta_1} \mathrm{e}^{\mathcal{L}^\beta_1t_1} \mathcal |A_1) \label{eq:chiqq}
\end{align}
where $\ell = \ell(\beta)$, \begin{equation}
\mathcal{L}^\beta_p := \mathcal{L} - \sum_{\lambda \in Y^Z_p(\beta)} \mathcal{L}_\lambda \label{eq:qqLbetap}
\end{equation}
with \begin{align}
Y^q_p(\beta) &:= \lbrace (q^\prime, k ): (k+1)2^{q^\prime-1}\ge  (k(\beta_p)+1)2^{q(\beta_p)-1} \rbrace, 
\end{align}

\end{lem}
\begin{proof}  The proof follows that of Lemma~\ref{lemma7}. First, we show that \begin{equation}
\prod_{q\in Z}\chi_q \cdot \mathrm{e}^{\CL t}|A_1) = \sum_{\lambda \in \mathcal{F}_Z} \sum_{m_0,\ldots, m_{\ell(\lambda)}=0}^\infty \frac{(t\mathcal{L})^{m_{\ell(\lambda)}} (t\mathcal{L}_{\lambda_{\ell(\lambda)}})(t\mathcal{L}^\lambda_{\ell(\lambda)})^{m_{\ell(\lambda)-1}}\cdots (t\mathcal{L}^\lambda_2)^{m_1} (t\mathcal{L}_{\lambda_1})(t\mathcal{L}^\lambda_1)^{m_0}}{(\ell(\lambda) + \sum_{j=0}^{\ell(\lambda)} m_j)!}   \label{eq:chiqqpart1}
\end{equation}
with $\mathcal{L}^\beta_p$ defined in (\ref{eq:qqLbetap}).  Every sequence on the right hand side of (\ref{eq:chiqqpart1}) correspond to a term on the left because each of these sequences contains a $\lambda \in \mathcal{F}_Z$ as a subsequence and hence has a long $q$-forward subsequence for each $q\in Z$. Next, every sequence on the left can be written as a sequence on the right: by construction, the $Y^Z_p(\beta)$ sets of couplings are chosen so that $\mathcal{L}^\beta_p$ does not change the irreducible $Z$-forward subsequence of the term.  The uniqueness of irreducible $Z$-forward subsequences also implies that every term on the right hand side shows up exactly once. As the coefficients of terms on both sides of (\ref{eq:chiqqpart1}) are the same, and we have found a bijection between the terms on both sides of the proposed equality (\ref{eq:chiqqpart1}), we have demonstrated its veracity.

Secondly, the generalized Schwinger-Karplus identity equates the right hand side of (\ref{eq:chiqqpart1})  to the right hand side of (\ref{eq:chiqq}).
\end{proof}

The remainder of the proof is entirely combinatorial.  As in (\ref{eq:mainequivclass2}), all quantum interference will now be hidden in the factors of $\mathrm{e}^{\mathcal{L}^\lambda_j t}$ in (\ref{eq:chiqq}).  We begin with the following lemma:

\begin{lem} 
\begin{equation}
\frac{\lVert \mathbb{P}_R \mathrm{e}^{\mathcal{L}t}|A_1)\rVert }{2 \lVert |A_1)\rVert} \le -1 + \exp\left[ \sum_{q=1}^{n_*} \left(\begin{array}{c} 2^{1-q}R \\ N_q  \end{array}\right)\frac{(2|t|)^{N_q}}{N_q!} \left( \sup_k \lVert \mathcal{L}_{(q,k)} \rVert \right)^{N_q}\right]. \label{eq:expsup}
\end{equation} 
\end{lem}

\emph{Proof.} We begin by combining (\ref{eq:sumZ}) and (\ref{eq:chiqq}): \begin{equation}
\lVert \mathbb{P}_R \mathrm{e}^{\mathcal{L}t}|A_1)\rVert = \left\lVert \mathbb{P}_R \sum_Z (-1)^{|Z|} \sum_{\beta \in \mathcal{F}_Z} \int\limits_{\mathrm{\Delta}^\ell(t)} \mathrm{d}t_\ell\cdots \mathrm{d}t_1 \mathrm{e}^{\mathcal{L}(t-t_\ell)} \mathcal{L}_{\beta_\ell} \mathrm{e}^{\mathcal{L}^\beta_\ell(t_\ell - t_{\ell-1})} \cdots \mathcal{L}_{\beta_1} \mathrm{e}^{\mathcal{L}^\beta_1t_1} \mathcal |A_1)\right\rVert 
\end{equation}
where $\ell := \ell(\beta)$.  Since for all individual couplings, $\mathcal{L}_{mn}$ is an antisymmetric superoperator, each $\mathcal{L}^\beta_p$ is antisymmetric, and $\mathrm{e}^{\mathcal{L}^\beta_ps}$ is orthogonal for any $s\in\mathbb{R}$.  Using Lemma \ref{renormalize}, we obtain \begin{align}
\lVert \mathbb{P}_R \mathrm{e}^{\mathcal{L}t}|A_1)\rVert &\le 2 \left\lVert \sum_Z (-1)^{|Z|} \sum_{\beta \in \mathcal{F}_Z} \int\limits_{\mathrm{\Delta}^\ell(t)} \mathrm{d}t_\ell\cdots \mathrm{d}t_1 \mathrm{e}^{\mathcal{L}(t-t_\ell)} \mathcal{L}_{\beta_\ell} \mathrm{e}^{\mathcal{L}^\beta_\ell(t_\ell - t_{\ell-1})} \cdots \mathcal{L}_{\beta_1} \mathrm{e}^{\mathcal{L}^\beta_1t_1} \mathcal |A_1)\right\rVert  \notag \\
&\le 2 \sum_Z \sum_{\beta \in \mathcal{F}_Z}  \int\limits_{\mathrm{\Delta}^\ell(t)} \mathrm{d}t_\ell\cdots \mathrm{d}t_1 \left\lVert  \mathrm{e}^{\mathcal{L}(t-t_\ell)} \mathcal{L}_{\beta_\ell} \mathrm{e}^{\mathcal{L}^\beta_\ell(t_\ell - t_{\ell-1})} \cdots \mathcal{L}_{\beta_1} \mathrm{e}^{\mathcal{L}^\beta_1t_1} \mathcal |A_1)\right\rVert \notag \\
&\le 2 \sum_Z \sum_{\beta \in \mathcal{F}_Z} \prod_{j=1}^\ell ( \lVert \mathcal{L}_{\beta_j} \rVert) \cdot \lVert |A_1)\rVert \int\limits_{\mathrm{\Delta}^{\ell}(t)} \mathrm{d}t_\ell\cdots \mathrm{d}t_1 \notag \\
&\le 2 \lVert |A_1)\rVert \sum_Z \sum_{\beta \in \mathcal{F}_Z} \frac{(|t|)^\ell}{\ell!} \prod_{j=1}^\ell \sup_k \;  \lVert \mathcal{L}_{(q(\beta_j),k)} \rVert  \label{eq:triangleexpsup}
\end{align}
where in the first line we used Proposition \ref{normdecrease}; in the second line we used the triangle inequality; in the third line we used the properties of Liouvillians described in Lemma \ref{renormalize} along with the fact that by construction each $\mathcal{L}_{(q,k)}$ in the irreducible sequence moves the operator to the right in such a way that we may use the effective norm from Lemma \ref{renormalize}, and in the fourth line we computed the volume of the simplex $\mathrm{\Delta}^{\ell}(t)$ as well as upper bounded $\lVert H_{\beta_j}\rVert$. 

Next, we count the number of irreducible $q$-forward sequences, which is simply the number of possible ways to choose $N_q$ different couplings out of $2^{1-q}R -1$ different choices of $k$: \begin{equation}
|\mathcal{F}_{\lbrace q\rbrace} | = \left(\begin{array}{c} 2^{1-q}R \\ N_q  \end{array}\right) \label{eq:Fq}
\end{equation}
To justify the factor of $2^{1-q}R-1$, observe that the maximal value of $k$ in $\mathcal{L}_{(q,k)}$ occurs when $(k+2)2^{q-1}=R$:  $k\le 2^{1-q}R-2$.  Since $k\ge 0$, we find $2^{1-q}R-1$ different values of $k$.

The irreducible $q$-forward subsequences of any irreducible $Z$-forward sequence $\lambda \in \mathcal{F}_Z$ are completely independent of each other.   Thus, the number of irreducible $Z$-forward sequences is given by product of the number of irreducible $q$-forward sequences for each $q\in Z$, together with the number of ways to weave together the few sequences:
\begin{equation}
|\mathcal{F}_Z| = \frac{(\sum_{q_1\in Z} N_{q_1})!}{\prod_{q_2\in Z} N_{q_2}!}\prod_{q\in Z} |\mathcal{F}_{\lbrace q\rbrace} |. \label{eq:numFZ}
\end{equation}
Since if $\beta\in\mathcal{F}_Z$, $\ell(\beta) = \sum_{q\in Z}N_q$, we can combine (\ref{eq:triangleexpsup}) and (\ref{eq:numFZ}) to obtain \begin{align}
\frac{\lVert \mathbb{P}_R \mathrm{e}^{\mathcal{L}t}|A_1)\rVert }{2\lVert |A_1)\rVert} &\le \sum_Z  \frac{(|t|)^{\sum_{q_1\in Z} N_{q_1}}}{(\sum_{q_1\in Z} N_{q_1})!}\frac{(\sum_{q_1\in Z} N_{q_1})!}{\prod_{q_2\in Z} N_{q_2}!}\prod_{q\in Z} \left( |\mathcal{F}_{\lbrace q\rbrace} | \left( \sup_k \lVert \mathcal{L}_{(q,k)} \rVert \right)^{N_q}\right) \notag \\
&\le \sum_Z \prod_{q\in Z} \left( |\mathcal{F}_{\lbrace q\rbrace} | \left( \sup_k \lVert \mathcal{L}_{(q,k)} \rVert\right)^{N_q} \frac{(|t|)^{N_q}}{N_q!}\right) \notag \\
&\le -1 + \prod_{q=1}^{n_*} \left[ 1 + |\mathcal{F}_{\lbrace q\rbrace} | \left( \sup_k \lVert \mathcal{L}_{(q,k)} \rVert\right)^{N_q} \frac{(|t|)^{N_q}}{N_q!}\right], \label{eq:lastexpsup}
\end{align}
where in the first two lines we made algebraic simplifications, and in the third line we used the distributive property together with the fact that there exist at least one scale with long q-forward sequence, i.e. $Z \in \mathbb{Z}_2^{\lbrace 1,\ldots, n_*\rbrace} - \emptyset$.  Combining (\ref{eq:Fq}) with (\ref{eq:lastexpsup}) and the elementary identity $1+x\le \mathrm{e}^x$ for any $x\in \mathbb{R}$, we obtain (\ref{eq:expsup}).
 \hfill \qedsymbol

The last step proving of Theorem 1 is simplifying the sum in the exponential of (\ref{eq:expsup}).
Plugging Lemma \ref{renormalize} into (\ref{eq:expsup}), we obtain \begin{align}
\frac{\lVert \mathbb{P}_R \mathrm{e}^{\mathcal{L}t}|A_1)\rVert }{2\lVert |A_1)\rVert} &\le -1 + \exp\left[ \sum_{q=1}^{n_*} \left(\begin{array}{c} 2^{1-q}R-1 \\ N_q  \end{array}\right)\frac{1}{N_q!} \left( \frac{2b|t|}{2^{q(\alpha^\prime-1)}} \right)^{N_q}\right]\\
&\le -1 + \exp\left[ \sum_{q=1}^{n_*} \frac{(2^{1-q}R)^{N_q}}{N_q!^2} \left( \frac{2b|t|}{2^{q(\alpha^\prime-1)}} \right)^{N_q}\right] \notag \\
&\le -1 + \exp\left[ \sum_{q=1}^{n_*} \left( \frac{R}{2^q N_q^2} \frac{4\mathrm{e}^2 b|t|}{2^{q(\alpha^\prime-1)}} \right)^{N_q}\right]  \label{eq:lemma9first}
\end{align}
where in the second line, we overestimated the choose function, and in the third line we used the inequality $ n!> (n/\mathrm{e})^n$ for any $n\in\mathbb{N}$. It is useful to determine the first value $q_*$ at which a long $q$-forward path has a single coupling: $N_q=1$ for $q\ge q_*$.    This occurs when \begin{equation}
\frac{M}{R}\ge \frac{1}{2^{1+q_*\alpha^\prime/2}},
\end{equation}
where we defined \begin{equation}
M = \sum_{q=1}^{n_*} 2^{-q(\alpha^\prime-2)/2}.
\end{equation}
Then, combining (\ref{eq:Nq}) and (\ref{eq:lemma9first}), 
we obtain \begin{equation}
\sum_{q=1}^{n_*} \left( \frac{R}{2^q N_q^2} \frac{4\mathrm{e}^2 b|t|}{2^{q(\alpha^\prime-1)}} \right)^{N_q} < \sum_{q=1}^{q_*-1} \left( 16\mathrm{e}^2 b|t| \frac{M^2}{R}  \right)^{N_q} + 4\mathrm{e}^2b|t| \sum_{q=q_*}^{n_*} \frac{R}{2^{q\alpha^\prime}} .
\end{equation}

We now analyze this sum for different ranges of $\alpha^\prime$.

\emph{Case 1: $\alpha^\prime>2$.}  In this regime, we begin by noting that 
\begin{equation}
N_1 > N_2 > \cdots > N_{q_*-1}.  \label{eq:Nqinequality}
\end{equation}
To derive this, note that the argument of the ceiling function in (\ref{eq:Nq}) changes by a factor of $2^{\alpha^\prime/2}$ each time $q$ changes by 1.  When $\alpha^\prime>2$, this factor is larger than 2, so once the argument is larger than 1, it changes by at least 1: $N_q \le N_{q-1}-1$.   Hence we may write \begin{equation}
\sum_{q=1}^{n_*} \left( \frac{R}{2^q N_q^2} \frac{4\mathrm{e}^2 b|t|}{2^{q(\alpha^\prime-1)}} \right)^{N_q} < \sum_{n=1}^{\infty} \left( 16\mathrm{e}^2 b|t| \frac{M^2}{R}  \right)^{n} + 4\mathrm{e}^2b|t| \sum_{q=q_*}^{n_*} \frac{R}{2^{q\alpha^\prime}}
\end{equation}
Next, we note that 
 \begin{equation}
M < \sum_{q=1}^\infty 2^{-q(\alpha^\prime-2)/2} = \frac{1}{1-2^{-(\alpha^\prime-2)/2}}.
\end{equation}
which implies that \begin{equation}
q_* \ge -1 + \frac{2}{\alpha^\prime}\log_2 \frac{R}{M} = -1 + \frac{2n_*}{\alpha^\prime} - \frac{2}{\alpha^\prime}\log_2 \frac{1}{1-2^{-(\alpha^\prime-2)/2}}.
\end{equation}
We conclude that \begin{equation}
\sum_{q=q_*}^{n_*} \frac{R}{2^{q\alpha^\prime}} < \frac{2^{\alpha^\prime}}{R(1-2^{-(\alpha^\prime-2)/2})^2} \sum_{n=0}^\infty 2^{-\alpha^\prime n} = \frac{2^{\alpha^\prime}}{R(1-2^{-(\alpha^\prime-2)/2})^2(1-2^{-\alpha^\prime})}
\end{equation}

\emph{Case 2: $1<\alpha^\prime<2$.}  In this regime, we must replace (\ref{eq:Nqinequality}) with the slightly weaker inequality \begin{equation}
N_1 > N_3 > N_5\cdots > N_{2\lceil q_*/2\rceil-1},
\end{equation}
because the argument of (\ref{eq:Nq}) now only varies by $2^{\alpha^\prime/2} \ge \sqrt{2}$ each time $q$ varies by 1.   Moreover, we now find \begin{equation}
M = \sum_{q=1}^{n_*} 2^{q(2-\alpha^\prime)/2} < R^{(2-\alpha^\prime)/2}\sum_{q^\prime=0}^\infty 2^{-q^\prime(2-\alpha^\prime)/2} = \frac{R^{(2-\alpha^\prime)/2}}{1-2^{-(2-\alpha^\prime)/2}}
\end{equation}
and that \begin{equation}
q_* = -1 + \frac{2}{\alpha^\prime}\log_2 \left(\left(1-2^{-(2-\alpha^\prime)/2}\right)R^{\alpha^\prime/2}\right) = -1 + n_* - \frac{2}{\alpha^\prime}\log_2 \frac{1}{1-2^{-(2-\alpha^\prime)/2}}.
\end{equation}
Hence, we obtain \begin{equation}
\sum_{q=1}^{n_*} \left( \frac{R}{2^q N_q^2} \frac{4\mathrm{e}^2 b|t|}{2^{q(\alpha^\prime-1)}} \right)^{N_q} < 2\sum_{n=1}^\infty \left( 16\mathrm{e}^2 b|t| \frac{M^2}{R}  \right)^{n} +  \frac{4\mathrm{e}^2b|t| }{1-2^{-\alpha^\prime}} \frac{2^{\alpha^\prime}}{(1-2^{-(2-{\alpha^\prime})/2})^{2}R^{\alpha^\prime-1}}
\end{equation}
where the 2 prefactor is a loose bound coming from that the $\sqrt{2}$ scaling - $N_1$ might equal to $N_2$.

\emph{Case 3: $\alpha^\prime=2$.}  In this regime, we obtain (\ref{eq:Nqinequality}), \begin{equation}
M = n_*,
\end{equation}and \begin{equation}
N_q = \left\lceil \frac{1}{2} \frac{R}{2^q \log_2 R} \right\rceil,
\end{equation}
implying that \begin{equation}
q_* \ge \log_2 \frac{R}{2\log_2 R}.
\end{equation}
Hence we may write \begin{equation}
\sum_{q=1}^{n_*} \left( \frac{R}{2^q N_q^2} \frac{4\mathrm{e}^2 b|t|}{2^{q(\alpha^\prime-1)}} \right)^{N_q} < \sum_{n=1}^\infty \left( 16\mathrm{e}^2 b|t| \frac{\log_2^2R}{R}  \right)^{n} + \frac{16}{3}\mathrm{e}^2b|t| \frac{4\log_2^2R}{R}
\end{equation}

 Each of the three cases leads to a simple bound. For simplicity in these final two paragraphs, we will take the values of $b$ calculated in frustrated models where $\alpha^\prime=\alpha$.   Analogous results hold for other models. As a function of time, we obtain 
 \begin{equation}
\frac{\lVert \mathbb{P}_R \mathrm{e}^{\mathcal{L}t}|A_1)\rVert }{2\lVert |A_1)\rVert}  \le \frac{c_1 t}{\mathcal{R}-c_1|t|} + c_2 \frac{|t|}{\mathcal{R}}
\end{equation}
where \begin{subequations}\begin{align}
\mathcal{R}(R) &= \left\lbrace\begin{array}{ll} R &\ \alpha > 2 \\ R \log^{-2}R &\ \alpha = 2 \\ R^{\alpha-1} &\ 1<\alpha<2 \end{array}\right., \\
c_1 &= b\cdot  \left\lbrace\begin{array}{ll} 16\mathrm{e}^2 (1-2^{-(\alpha-2)/2})^{-2} &\  \alpha > 2 \\ 16\mathrm{e}^2 &\ \alpha = 2 \\  32\mathrm{e}^2 (1-2^{-(2-\alpha)/2})^{-2}  &\ 1<\alpha<2 \end{array}\right., \\
c_2 &= b\cdot \left\lbrace\begin{array}{ll} 2^{2+\alpha}\mathrm{e}^2 (1-2^{-\alpha})^{-1}(1-2^{-(\alpha-2)/2})^{-2} &\  \alpha > 2 \\ \frac{64}{3}\mathrm{e}^2 &\ \alpha = 2 \\  2^{2+\alpha}\mathrm{e}^2 (1-2^{-\alpha})^{-1}  (1-2^{-(2-\alpha)/2})^{-2}  &\ 1<\alpha<2 \end{array}\right., 
\end{align}\end{subequations}and $c_{1,2}$ are O(1) constant.      Now observe that \begin{equation}
\frac{\lVert \mathbb{P}_R \mathrm{e}^{\mathcal{L}t}|A_1)\rVert }{2\lVert |A_1)\rVert}  \le  (2c_1+c_2) \frac{|t|}{\mathcal{R}}, \;\;\;\; \left(|t|<\frac{\mathcal{R}}{2c_1}\right). \label{eq:lastinequality}
\end{equation}

Recall the definition of the scrambling time $t_{\mathrm{s}}^\delta(R)$ from (\ref{eq:scrambling}).   Using Proposition \ref{CommToProj}, we conclude that 
\begin{equation}
\frac{\delta}{2} \le  2 (2c_1+c_2)\frac{t_{\mathrm{s}}^\delta(R)}{\mathcal{R}}
\end{equation}
Since $\delta < 2$, the right hand side becomes larger than 1 before the inequality (\ref{eq:lastinequality}) breaks down.  Since $\frac{1}{2}R\ge r\ge R$, we conclude that \begin{equation}
t_{\mathrm{s}}^\delta(r) \ge\frac{\delta}{2 b}\cdot \left\lbrace \begin{array}{ll}\dfrac{(1-2^{-\alpha})(1-2^{-(2-\alpha)/2})^2)}{(32+2^{2+\alpha})\mathrm{e}^2} r &\  \alpha > 2 \\ \dfrac{3}{160\mathrm{e}^2}  \dfrac{r}{\log^2 r}&\ \alpha = 2 \\  \displaystyle \dfrac{(1-2^{-\alpha})(1-2^{-(2-\alpha)/2})^2)}{(64+2^{2+\alpha})\mathrm{e}^2}  r^{\alpha-1} &\ 1<\alpha<2   \end{array}\right..
\end{equation}
This proves the main Theorem 1.
 \hfill \qedsymbol
\end{document}